\documentclass[noshowpacs,pre,preprintnumbers,superscriptaddress,amsmath,amssymb,floatfix]{revtex4}
\usepackage[caption=false]{subfig}
\usepackage{graphicx}
\usepackage{color}
\usepackage{placeins}
\usepackage[normalem]{ulem}
\usepackage[colorlinks=true,citecolor=blue,linkcolor=blue,urlcolor=blue]{hyperref}
\usepackage{physics}

\usepackage{amsthm}
\usepackage{mathtools}

\newtheorem{theorem}{Theorem}

\usepackage[caption=false]{subfig}

\begin{document}

\title{Leggett-Garg-like Inequalities from a Correlation Matrix Construction}
    \author{Dana Ben Porath}
    \thanks{Formerly Dana Vaknin}
    \affiliation{Faculty of Engineering and the Institute of Nanotechnology and Advanced Materials, Bar Ilan University, Ramat Gan, Israel}
    \author{Eliahu Cohen}
    \email[]{eliahu.cohen@biu.ac.il}
	\affiliation{Faculty of Engineering and the Institute of Nanotechnology and Advanced Materials, Bar Ilan University, Ramat Gan, Israel}
    \date{\today}

\begin{abstract}
The Leggett-Garg Inequality (LGI) constrains, under certain fundamental assumptions, the correlations between measurements of a quantity $Q$ at different times. 
Here we analyze the LGI, and propose similar but somewhat more elaborate inequalities, employing a technique that utilizes the mathematical properties of correlation matrices, which was recently proposed in the context of nonlocal correlations. 
We also find that this technique can be applied to inequalities that combine correlations between different times (as in LGI)  and correlations between different locations (as in Bell inequalities). All the proposed bounds include additional correlations compared to the original ones and also lead to a particular form of complementarity. A possible experimental realization and some applications are briefly discussed.
\end{abstract}

\maketitle

\section{Introduction}

Leggett and Garg, in their seminal work \cite{legget-prl1985}, provided constraints on the correlations between measurements of a single quantity at different times.
They showed that given the definition $C_{ij}\equiv \langle Q_i Q_j \rangle$ for the correlation between two measurements of a quantity $Q$ at times $t_i$ and $t_j$, the sum $|C_{12} + C_{23} + C_{34} - C_{14}|$ is bounded by 2 in any scenario that maintains ``macrorealism'' and ``non-invasive measurability'' \cite{emary-iop2013}.
Such a scenario represents the classical physics view of a macroscopic system, as a system that cannot be in two or more states at the same time, and in which it is possible to measure the state with only an arbitrarily small perturbation of it.
Determining whether Leggett-Garg Inequality (LGI) holds for a given system, assists in distinguishing between systems that obey this classical view and those that exhibit non-classical behavior (in the particular sense linked to  macrorealism).
If a system violates the LGI it necessarily exhibits non-classical behavior. However, recent works have shown that the contrary is not always true, i.e., a system can satisfy the LGI but still violate macroreaslism \cite{majidy-pra2019,majidy-pra2021}.
The initial motivation for the LGI was using this method to determine whether quantum coherence appears even in macroscopic systems.
Later works focused on experimentally finding LGI violations in microscopic quantum systems using various measurement types such as ideal negative measurements \cite{knee-naturecomm2012}.
A recent work by Shenoy {\it et al.} \cite{shenoy-pra2017} presents the applicability of LGI in the area of quantum cryptography. They showed how the amount of violation of the LGI indicates that a hacking attempt was made during the quantum key distribution protocol \cite{bb84-arxiv2020}.

The LGI relates to bounds on correlations of a single quantity measured typically on the same system at different times by the same party, in a manner which mathematically resembles the well-known Clauser-Horne-Shimony-Holt (CHSH) inequalities, providing bounds on correlations between quantities measured typically on a bipartite system at two different locations by two different parties \cite{bell-physics1964,clauser-prl1969,Cirel'son-letters1980}. 
Similarly to the ``classical'' LGI bound of 2, previous works have shown that under quantum assumptions, the same correlations are bounded by $2\sqrt{2}$ \cite{budroni-prl2013}.
Here we find more informative LGI bounds, tighter than $2\sqrt{2}$, by studying a more general definition for the correlations, which under certain conditions coincides with the common LGI correlations for quantum systems \cite{fritz-njp2010}. 
We apply a mathematical method that was recently proposed for finding richer bounds for CHSH and other novel inequalities \cite{carmi2018significance,carmi-scienceadvances2019,carmi-njp2019,te2019multiplicative}. 

We show that maximal violations of the classical LGI bound in the newly found inequalities, directly follow from the absence of correlations between an operator and itself at certain times.
On the practical level, our results may allow one to better design quantum temporal correlations, but fundamentally this highlights the necessity of certain non-commutativity between the measured operator and the Hamiltonian (see also \cite{cohen2020praise} for the broader consequences of non-commutativity and uncertainty).

In addition to their theoretical merits, we suggest that these tighter bounds may be beneficial for applied purposes as well, e.g. for devising quantum key distribution protocols (similarly to \cite{shenoy-pra2017}) or analyzing quantum metrological schemes (similarly to \cite{frowis-prl2016}).

For completeness, we show in the Appendix that a different definition of the temporal correlation leads to other bounds.

\section{Materials and Methods}

We define the generalized correlation function $C(X,Y)$ --- for any two Hermitian operators $X$ and $Y$ --- as 

\begin{equation}
    C(X,Y) = \frac{\frac{1}{2}\langle \left \{X, Y\right \} \rangle - \langle X \rangle \langle Y \rangle}{\Delta_{X}\Delta_{Y}},
    \label{eq:generalized_corr}
\end{equation}
where $\{X,Y\}$ denotes the anti-commutator and $\Delta_X = \sqrt{\langle X^2 \rangle - \langle X \rangle^2}$, $\Delta_Y = \sqrt{\langle Y^2 \rangle - \langle Y \rangle^2}$ are the standard deviations.
From the Schr\"odinger Robertson uncertainty relation \cite{robertson-pr1929,schrodinger-1930}, it is easy to prove that this correlation is bounded between -1 to 1.
If $X$ and $Y$ represent projective quantum measurements of values $\pm 1$ with an expected value of 0, this correlation coincides with the symmetric correlation $\frac{1}{2}\langle \left \{X, Y\right \} \rangle$, which Fritz proposed as a quantum analog of the LGI correlation \cite{fritz-njp2010}.

The motivation behind this generalized correlation definition is twofold: except for the aforementioned fact that in the case of projective measurements it generalizes the standard symmetric correlation used for LGI, it is also widely used in quantum optics \cite{adesso-jpa2007,pirandola-pra2009} and can be therefore readily generalized to the case of continuous variables.

After constructing the correlations matrix we will employ semi-positive definiteness conditions, similar to those in \cite{carmi2018significance,carmi-scienceadvances2019,carmi-njp2019,te2019multiplicative}, to derive our Leggett-Garg-like inequalities.

\section{Results}

Here we present four theorems and corresponding proofs entailing analytical bounds for correlations between measurements of a quantum system at different times. 
Theorems \ref{theory:th1}, \ref{theory:th2}, and \ref{theory:th3} discuss elaborate Leggett-Greg-like inequalities describing constraints on generalized correlations of measurements at four consecutive measurement times.
The measurement of a quantity $Q$ at time $t_i$ is described by the Hermitian operator $Q_i$.
Theorem \ref{theory:th1} is significant mainly because it provides a tighter bound than $2\sqrt{2}$, which is the known bound under typical quantum assumptions.
Theorem \ref{theory:th2} resembles the TLM inequality which is a significant bound that was derived independently by Tsirelson, Landau and Masanes \cite{tsirelson-jsm1987,landau-foundationsofphysics1988,masanes-arxiv2003}. The structure of TLM inequality is characterized by the fact it bounds products of correlations and not merely their sums as in other inequalities. 
Theorem \ref{theory:th3} demonstrates a complementarity relation between all pairs of correlations of the four measurements. 

Theorem \ref{theory:th4} relates to a combination of CHSH and LGI \cite{dressel-pra2014,white-npj2016}, where the constraints are on generalized correlations between quantities measured at two different times $t_1$ and $t_2$ by two parties, Alice and Bob, each in a different location.
The two consecutive measurements of Alice (Bob) are represented by the Hermitian operators $A_1$ and $A_2$ ($B_1$ and $B_2$).

\begin{theorem}
{\bf Elaborate Leggett-Garg-like inequality.}  

Given four consecutive measurements, define the generalized LGI parameter as

\begin{equation}
    L = |C(Q_1,Q_2) + C(Q_2,Q_3) + C(Q_3,Q_4) - C(Q_1,Q_4)|.
    \label{eq:L_def}
\end{equation}
The following holds

\begin{equation}
    L \leq 2\sqrt{1+\sqrt{1-\max{\left\{C(Q_1,Q_3)^2,C(Q_2,Q_4)^2\right\}}}}.
    \label{eq:th1_bound}
\end{equation}

\label{theory:th1}
\end{theorem}

\begin{proof} [Proof of Theorem 1]

Let $C$ be the following correlation matrix 

\begin{gather}
C = 
\begin{pmatrix}
C(Q_{2i},Q_{2i}) & C(Q_{2i},Q_{1}) & C(Q_{2i},Q_{3})\\
C(Q_{2i},Q_{1}) & C(Q_{1},Q_{1}) & C(Q_{1},Q_{3})\\
C(Q_{2i},Q_{3}) & C(Q_{1},Q_{3}) & C(Q_{3},Q_{3})
\end{pmatrix}.
\end{gather}
where $i = 1,2$. $C$ is a positive semi-definite matrix, i.e., $C\succeq 0$ (see \cite{carmi-njp2019,carmi-scienceadvances2019} for more details regarding the construction and properties of such matrices). 
Therefore, by the Schur complement condition for positive semi-definiteness, 

\begin{gather}
\begin{pmatrix}
1 & C(Q_{1},Q_{3})\\
C(Q_{1},Q_{3}) & 1
\end{pmatrix}
\succeq 
\begin{pmatrix}
C(Q_{2i},Q_{1})\\
C(Q_{2i},Q_{3})
\end{pmatrix}
\begin{pmatrix}
C(Q_{2i},Q_{1}) & C(Q_{2i},Q_{3})
\end{pmatrix}.
\label{eq:shur_bound}
\end{gather}
Let $v_j^T = ((-1)^j, 1)$.
Multiplying by $v_j^T$ from the left and $v_j$ from the right, the above inequality implies

\begin{equation}
    2[1+(-1)^j C(Q_1,Q_3)]\geq [ C(Q_{2i},Q_3)+(-1)^j C(Q_{2i},Q_1) ]^2.
\end{equation}
For $j = i-1$,

\begin{align}
    \begin{split}
    &[C(Q_2,Q_1)+C(Q_2,Q_3)]^2\leq 2 \left[1+C(Q_1,Q_3) \right] \\
    &[C(Q_4,Q_3)-C(Q_4,Q_1)]^2\leq 2 \left[1-C(Q_1,Q_3) \right].
    \end{split}
    \label{eq:two_bounds_13}
\end{align}
Since $C(Q_i,Q_j) = C(Q_j,Q_i)$ for any $i$ and $j$, then

\begin{align}
    \begin{split}
    &|C(Q_1,Q_2)+C(Q_2,Q_3)|\leq \sqrt{2 \left[1+C(Q_1,Q_3) \right]} \\
    &|C(Q_3,Q_4)-C(Q_1,Q_4)|\leq \sqrt{2 \left[1-C(Q_1,Q_3) \right]}.
    \end{split}
    \label{eq:two_bounds_13_final}
\end{align}
Using the triangle inequality on the two expressions in the l.h.s of Eq. \eqref{eq:two_bounds_13_final}, we derive the following 

\begin{equation}
    L\leq 2\sqrt{1+\sqrt{1-C(Q_1,Q_3)^2}}.
    \label{eq:one_final_bound_13}
\end{equation}

By repeating the analytical derivation above for the following correlation matrix 

\begin{gather}
\tilde{C} = 
\begin{pmatrix}
C(Q_{2i-1},Q_{2i-1}) & C(Q_{2i-1},Q_{4}) & C(Q_{2i-1},Q_{2})\\
C(Q_{2i-1},Q_{4}) & C(Q_{4},Q_{4}) & C(Q_{4},Q_{2})\\
C(Q_{2i-1},Q_{2}) & C(Q_{4},Q_{2}) & C(Q_{2},Q_{2})
\end{pmatrix},
\end{gather}
the multiplications by $v_j^T$ and $v_j$ gives rise to:
 
\begin{equation}
    2[1+(-1)^j C(Q_4,Q_2)]\geq [ C(Q_{2i-1},Q_2)+(-1)^j C(Q_{2i-1},Q_4) ]^2.
    \label{eq:one_bound_24}
\end{equation}
From Eq. \eqref{eq:one_bound_24}, for the cases $i = j = 1$ and $i = j = 2$, the following inequality is derived using the triangle inequality

\begin{equation}
    L\leq 2\sqrt{1+\sqrt{1-C(Q_2,Q_4)^2}}.
    \label{eq:one_final_bound_24}
\end{equation}
Finally, Theorem \ref{theory:th1} follows from Eqs. \eqref{eq:one_final_bound_13} and \eqref{eq:one_final_bound_24}. 
Thus, we prove that given our assumptions, the generalized LGI parameter has a bound which is tighter than $2\sqrt{2}$, and $2\sqrt{2}$ can be reached only if the correlations $C(Q_1,Q_3)$ and $C(Q_2,Q_4)$ are equal to $0$. 
This result generalizes the known bound of $2\sqrt{2}$ for the LGI with projective measurements of values $\pm 1$ \cite{budroni-prl2013}. 
\end{proof}

\begin{theorem}
{\bf Leggett-Garg-like inequality in the TLM form.}

Given four consecutive measurements, 

\begin{align}
    \begin{split}
    &|C(Q_2,Q_1) C(Q_2,Q_3) - C(Q_4,Q_1) C(Q_4,Q_3)| \leq \\
    &\sqrt{(1-C(Q_2,Q_1)^2)(1-C(Q_2,Q_3)^2)}+ 
    \sqrt{(1-C(Q_4,Q_1)^2)(1-C(Q_4,Q_3)^2)}.
    \end{split}
    \label{eq:tlm}
\end{align}

\label{theory:th2}
\end{theorem}

\begin{proof} [Proof of Theorem 2]

Eq. \eqref{eq:shur_bound} implies

\begin{gather}
\begin{pmatrix}
1 - C(Q_{2i},Q_{1})^2 & C(Q_{1},Q_{3}) - C(Q_{2i},Q_{1}) C(Q_{2i},Q_{3})\\
C(Q_{1},Q_{3}) - C(Q_{2i},Q_{3}) C(Q_{2i},Q_{1}) & 1 - C(Q_{2i},Q_{3})^2
\end{pmatrix}
\succeq 
0.
\label{eq:shur_bound_t2}
\end{gather}
The determinant of the above matrix is non-negative, and thus,

\begin{equation}
     |C(Q_{1},Q_{3}) - C(Q_{2i},Q_{1}) C(Q_{2i},Q_{3})|\leq
     \sqrt{(1 - C(Q_{2i},Q_{1})^2)(1 - C(Q_{2i},Q_{3})^2)}.
\end{equation}
For the cases $i = 1$ and $i = 2$, we obtain the following inequalities, respectively,

\begin{align}
    \begin{split}
    &|C(Q_{1},Q_{3}) - C(Q_{2},Q_{1}) C(Q_{2},Q_{3})|\leq
     \sqrt{(1 - C(Q_{2},Q_{1})^2)(1 - C(Q_{2},Q_{3})^2)} \\
    &|C(Q_{1},Q_{3}) - C(Q_{4},Q_{1}) C(Q_{4},Q_{3})|\leq
     \sqrt{(1 - C(Q_{4},Q_{1})^2)(1 - C(Q_{4},Q_{3})^2)}.
    \end{split}
    \label{eq:two_bounds_t2}
\end{align}
Finally, Theorem \ref{theory:th2} is derived from the triangle inequality and Eq. \eqref{eq:two_bounds_t2}.
The resulting inequality presents the TLM criterion for correlations between measurements at different times.
\end{proof}

\begin{theorem}
{\bf Leggett-Garg-like inequality in the form of a complementarity relation.}

Given four consecutive measurements,
    \begin{equation}
    \left(\frac{L}{2\sqrt{2}}\right)^2+\left(\frac{C(Q_1,Q_3)}{2\sqrt{2}}\right)^2+\left(\frac{C(Q_2,Q_4)}{2\sqrt{2}}\right)^2 \leq 1.
    \label{eq:th3}
    \end{equation}
\label{theory:th3}
\end{theorem}

\begin{proof} [Proof of Theorem 3]

From Eq. \eqref{eq:one_final_bound_13},

\begin{equation}
    L^2\leq 4\left( 1 + \sqrt{1-C(Q_1,Q_3)^2} \right).
\end{equation}
Since $\sqrt{1-a} \leq 1 - a/2$ for $a\in [0,1]$, then

\begin{equation}
L^2 + 2C(Q_1,Q_3)^2 \leq 8
\label{eq:t3_13}
\end{equation}
and similarly we can derive from Eq. \eqref{eq:one_final_bound_24},

\begin{equation}
L^2 + 2C(Q_2,Q_4)^2 \leq 8
\label{eq:t3_42}
\end{equation}
Theorem \ref{theory:th3} is derived after summing Eqs. \eqref{eq:t3_13} and \eqref{eq:t3_42}.
This inequality demonstrates a complementarity relation between all six correlations of pairs from the four measurements.
\end{proof}

\begin{theorem}
{\bf Elaborate Bell-Leggett-Garg-like inequality.}

For two consecutive measurements $A_1$ and $A_2$ of Alice, and two consecutive measurements $B_1$ and $B_2$ of Bob, define the generalized Bell-Leggett-Garg inequality parameter as 

\begin{equation}
    BLG = |C(A_1,A_2)+C(A_1,B_2)+C(B_1,B_2)-C(B_1,A_2)|.
\end{equation}
The following holds

\begin{equation}
    BLG \leq 2\sqrt{1+\sqrt{1-\max{\left\{C(A_1,B_1)^2,C(A_2,B_2)^2\right\}}}}.
\end{equation}

\label{theory:th4}
\end{theorem}

\begin{proof} [Proof of Theorem 4]

Let $C_X$ be the following correlation matrix
\begin{gather}
C_X = 
\begin{pmatrix}
C(X,X) & C(X,A_2) & C(X,B_2)\\
C(X,A_2) & C(A_2,A_2) & C(A_2,B_2)\\
C(X,B_2) & C(A_2,B_2) & C(B_2,B_2)
\end{pmatrix},
\end{gather}
for $X \in \{A_1,B_1\}$.
Following the analysis in the proof of Theorem \ref{theory:th1}, i.e.,  using the Schur complement condition for positive semi-definiteness and after multiplying by $v_j^T = ((-1)^j, 1)$ and $v_j$, we obtain

\begin{equation}
    | C(X,B_2)+(-1)^j C(X,A_2) | \leq
    \sqrt{2[1+(-1)^j C(A_2,B_2)]}.
\end{equation}
For the two cases, $(X = A_1~\&~ j = 0)$ and $(X = B_1~\&~ j = 1)$, we obtain the following inequalities, respectively,

\begin{align}
    \begin{split}
    &| C(A_1,B_2)+ C(A_1,A_2) | \leq
    \sqrt{2[1+ C(A_2,B_2)]} \\
    &| C(B_1,B_2) - C(B_1,A_2) | \leq
    \sqrt{2[1 - C(A_2,B_2)]}.
    \end{split}
    \label{eq:two_bounds_t4}
\end{align}
From the triangle inequality,

\begin{equation}
    BLG\leq 2\sqrt{1+\sqrt{1-C(A_2,B_2)^2}}.
    \label{eq:t4_one_final_bound_2}
\end{equation}

Similarly, by replacing the variables $A_2$ and $B_2$ by $A_1$ and $B_1$, respectively,

\begin{equation}
    BLG\leq 2\sqrt{1+\sqrt{1-C(A_1,B_1)^2}}.
    \label{eq:t4_one_final_bound_1}
\end{equation}
Finally, Theorem \ref{theory:th4} is derived from Eqs. \eqref{eq:t4_one_final_bound_2} and \eqref{eq:t4_one_final_bound_1}.
\end{proof}

\subsection*{An example of a system that upholds our new bounds}

Here we demonstrate Theorems \ref{theory:th1},\ref{theory:th2} and \ref{theory:th3} for a specific spin model \cite{emary-iop2013,halliwell-pra2016}, which is defined by the following Hamiltonian and observable

\begin{align}
    \begin{split}
    &H = \frac{\hbar \omega}{2}\sigma_x = \frac{\hbar \omega}{2}
        \begin{pmatrix}
        0 & 1 \\
        1 & 0 
        \end{pmatrix}\\
    &Q = \sigma_z = 
        \begin{pmatrix}
        1 & 0 \\
        0 & -1 
        \end{pmatrix}.
    \end{split}
    \label{eq:example_system}
\end{align}
The time evolution of $Q$ is

\begin{gather}
Q_t = 
\begin{pmatrix}
\cos(\omega t) & -i \sin(\omega t) \\
i \sin(\omega t) & -\cos(\omega t) 
\end{pmatrix},
\end{gather}
and therefore, according to Eq. \eqref{eq:generalized_corr}, the generalized correlation between two measurements of $Q$ at times $t$ and $s$ is

\begin{equation}
    C(Q_t,Q_s) = \cos{(\omega (t-s))}.
\end{equation}

To demonstrate Theorem \ref{theory:th1}, for any four consecutive measurement times, define $D1$ as the difference between our bound (the r.h.s of Eq. \eqref{eq:th1_bound}) and the LGI parameter (Eq. \eqref{eq:L_def}). Thus, in our system,

\begin{align}
    \begin{split}
    D1 = &~2\sqrt{1+\sqrt{1-\max{\left\{\cos^2{(\omega (t_1-t_3))},\cos^2{(\omega (t_4-t_2))}\right\}}}}
    \\ & - 
    |\cos{(\omega (t_1-t_2))}+\cos{(\omega (t_2-t_3))}+\cos{(\omega (t_3-t_4))}-\cos{(\omega (t_1-t_4))}|.
    \end{split}
\end{align}
Similarly, to demonstrate Theorem \ref{theory:th2}, define $D2$ as the difference between the r.h.s and the l.h.s of Eq.\eqref{eq:tlm}, and thus,
\begin{align}
    \begin{split}
    D2 = &~|\sin(\omega (t_2-t_1))\sin(\omega (t_2-t_3))|+|\sin(\omega (t_4-t_1))\sin(\omega (t_4-t_3))|
    \\ & - 
    |\cos(\omega (t_2-t_1))\cos(\omega (t_2-t_3))-
    \cos(\omega (t_4-t_1))\cos(\omega (t_4-t_3))|.
    \end{split}
\end{align}
It can be shown numerically that $D1,D2 \geq 0$ for any $t_1, t_2, t_3$ and $t_4$ (see  Figs.~\ref{fig:Th1_Th2_Th3}(a)-(b) for a certain range of parameters).
An example of a measurement time series in which the bound in Theorem \ref{theory:th1} and the LGI parameter are both equal to $2\sqrt{2}$ is $t_1 = 0, t_2 = \pi/4, t_3 = \pi/2$ and $t_4 = 3\pi/4$.

In Fig.~\ref{fig:Th1_Th2_Th3}(c), we demonstrate Theorem \ref{theory:th3} by showing that the l.h.s of Eq. \eqref{eq:th3} is indeed smaller or equal to 1. To do so, we display all possible data points using the axes $L/2\sqrt{2}$, $C(Q_1,Q_3)/2\sqrt{2}$ and $C(Q_2,Q_4)/2\sqrt{2}$, and note that they all reside within the unit sphere. 

\begin{figure*}
	\centering
	\subfloat[Theorem \ref{theory:th1} for our example.]{\includegraphics[width=0.33\linewidth]{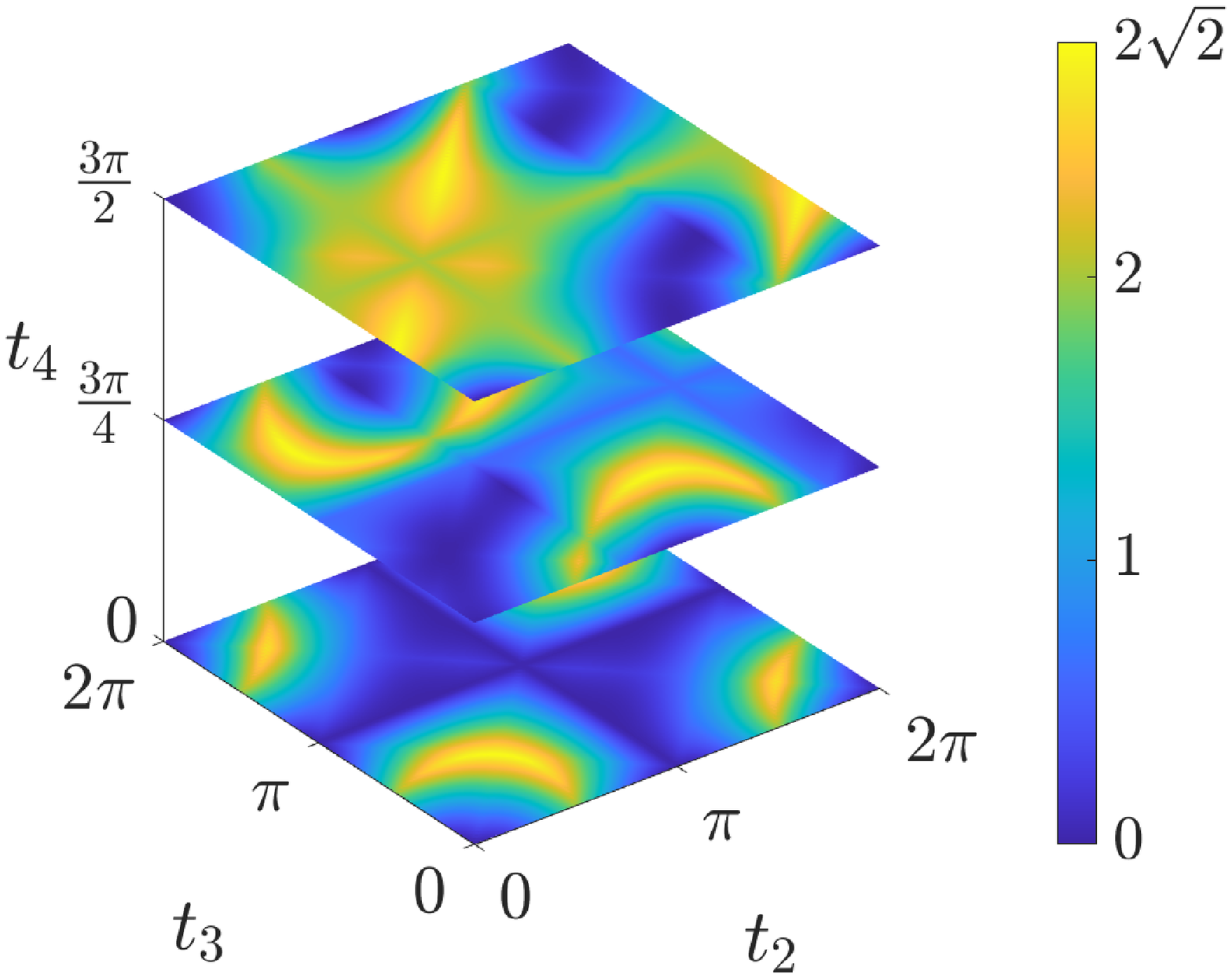}}
            \hfil
	\subfloat[Theorem \ref{theory:th2} for our example.]{\includegraphics[width=0.33\linewidth]{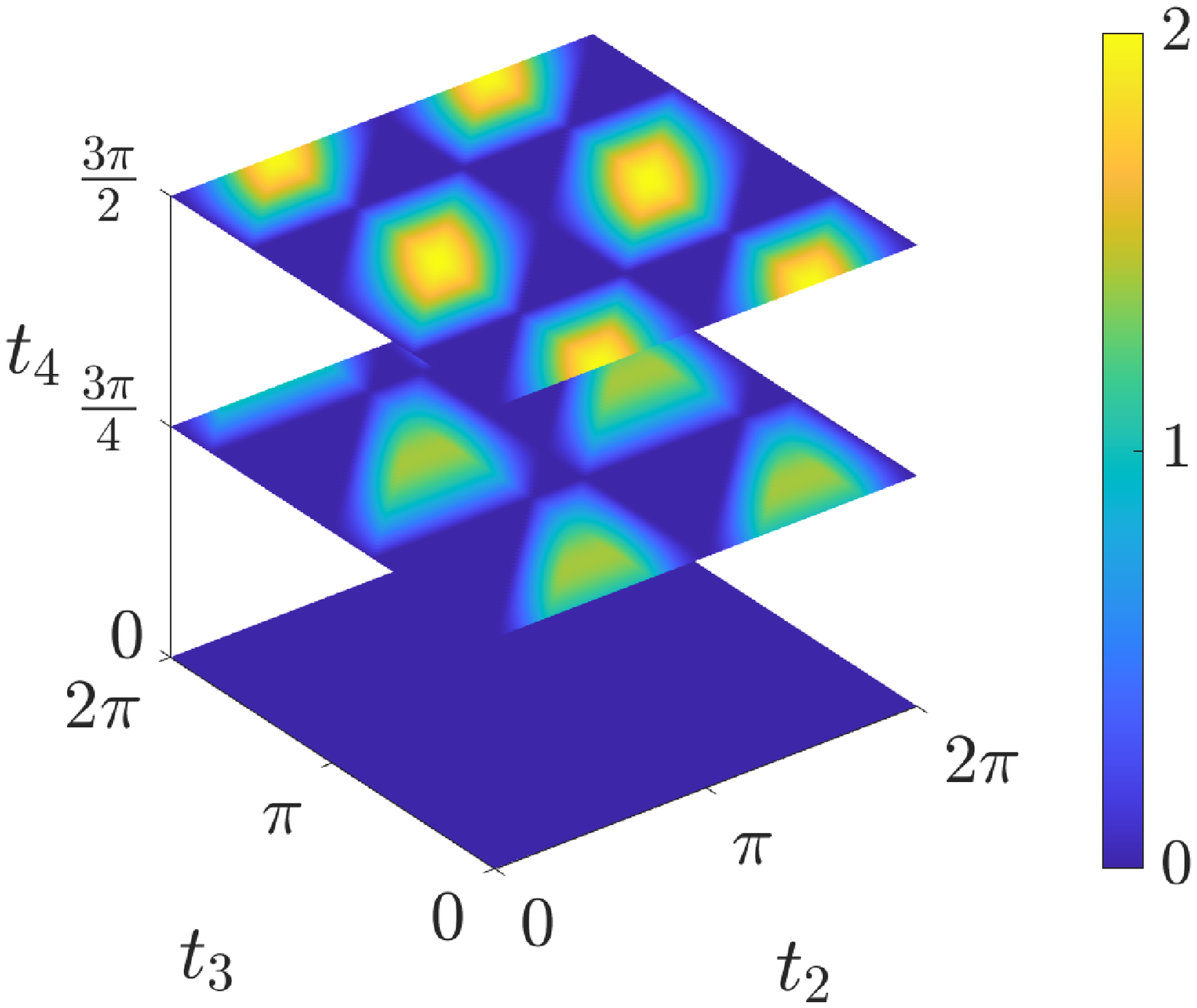}}
            \hfil
	\subfloat[Theorem \ref{theory:th3} for our example.]{\includegraphics[width=0.33\linewidth]{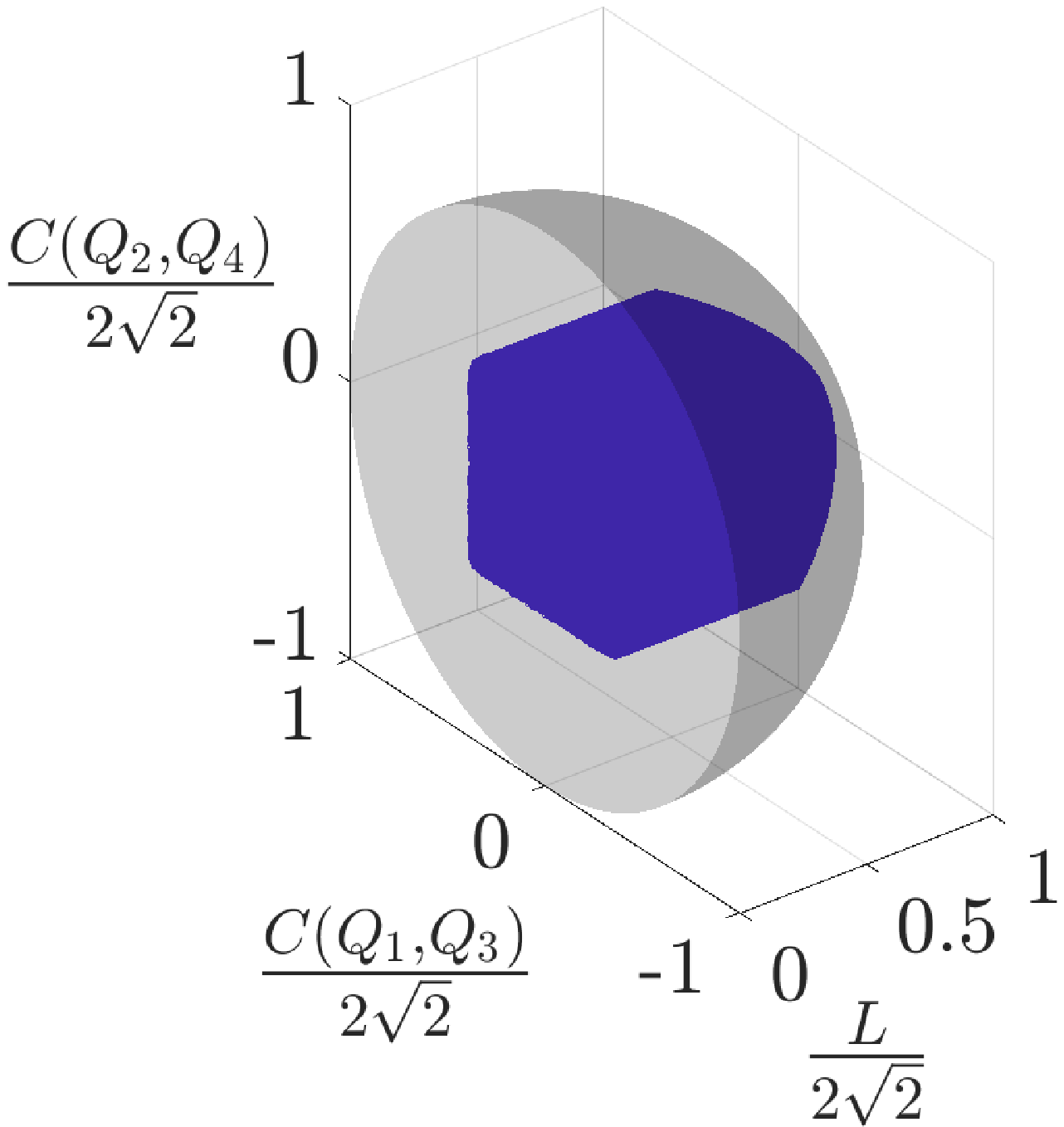}}
        \caption{\textbf{Demonstration of Theorems \ref{theory:th1}, \ref{theory:th2} and \ref{theory:th3} for our spin model}.
        \textbf{(a) - (b)}, the values of $D1$ and $D2$ respectively, both of which have a minimal value of  0, indicating Eqs. \eqref{eq:th1_bound} and \eqref{eq:tlm} are upheld.
        \textbf{(c)} The indigo area represents all data points, while the grey area represents the bounds of half of the unit sphere. All indigo points are within the gray area, indicating Eq. \eqref{eq:th3} is upheld.
        }
        \label{fig:Th1_Th2_Th3}
\end{figure*}

\section{Discussion}

In this manuscript, we have studied generalizations of the LGI. We have utilized a framework based on the quantum correlation matrix \cite{carmi2018significance,carmi-scienceadvances2019,carmi-njp2019} which yielded now a more detailed version of the LGI by incorporating additional correlations compared to the standard case. The results emphasize the major role of the positive-semidefinite correlation matrix in quantum mechanics, not only in ``spatial'' scenarios but also in ``temporal'' ones. They also demonstrate a type of complementarity --- in order for certain correlations to achieve their maximal values (left hand side of Eq. \eqref{eq:th1_bound}) others must vanish (those at the right hand side of Eq. \eqref{eq:th1_bound}). 

We suggest that apart from its foundational and theoretical merits, the proposed bound may help in designing temporal correlations as well as the dynamics giving rise to them. They may also have practical implications in the area of quantum cryptography, as a possible generalization to existing encryption protocols that are based on LGI \cite{shenoy-pra2017} or in the field of quantum metrology, possibly assisting approaches such as \cite{frowis-prl2016}, or in LGI-based quantum computation assessment \cite{vitale-scientificreports2019,santini-pra2022}.
In addition, since there are known connections between correlation matrices and both classical \cite{carmi-entropy2018} and quantum \cite{liu-physicsa2020} Fisher information, the bounds derived in this manuscript can provide analogous bounds on elements of the Fisher information matrix.

Finally, the work presented in this manuscript provides concrete and measurable predictions and therefore it can be verified in direct experiments. The system we provide as an example in the Results section is equivalent to the system which was experimentally measured in \cite{athalye-prl2011}, and the experimental results there are consistent with our bounds. 
In the Appendix, we propose an additional definition for the correlations and derive appropriate LGI-like bounds, which were not directly measured in previous experiments but can be either calculated or measured, e.g. via weak measurements \cite{aharonov-prl1988}.

\section*{Acknowledgements}
E.C. was supported by the Israeli Innovation Authority under Project 73795 and the Eureka program, by Elta Systems Ltd., by the Pazy Foundation, by the Israeli Ministry of Science and Technology, and by the Quantum Science and Technology Program of the Israeli Council of Higher Education.
The authors acknowledge many helpful discussions with Avishy Carmi and Rain Lenny.

\section*{Appendix}

In this section, we discuss two subjects --- another correlation definition and the significance of the underlying assumptions.

In principle, one can generalize the above results to the case of non-Hermitian quantum mechanics \cite{moiseyev2011non}. By defining the correlation $C(Q_i,Q_j)$ as the following complex correlation coefficient \cite{carmi-njp2019,carmi-scienceadvances2019}, 

\begin{equation}
    C(Q_i,Q_j) = \frac{\langle Q_i Q_j^\dagger \rangle - \langle Q_i \rangle \langle Q_j \rangle^\dagger} {\Delta_{Q_i}\Delta_{Q_j}},
\end{equation}
and the complex-valued LGI parameter as 

\begin{equation}
    L = |C(Q_1,Q_2) + C(Q_2,Q_3) + C(Q_3,Q_4) - C(Q_1,Q_4)|,
\end{equation}
we obtain 

\begin{align}
    \begin{split}
    & L \leq 2\sqrt{1+\sqrt{1-\max{\left\{\Re^2[C(Q_1,Q_3)],\Re^2[C(Q_4,Q_2)]\right\}}}} \\
    & \left(\frac{L}{2\sqrt{2}}\right)^2 + \left(\frac{\Re[C(Q_1,Q_3)]}{2\sqrt{2}}\right)^2  + 
    \left(\frac{\Re[C(Q_4,Q_2)]}{2\sqrt{2}}\right)^2 
    \leq 1,
    \end{split}
\end{align}
if $C(Q_2,Q_3) = C(Q_3,Q_2)$ or $[C(Q_1,Q_2) = C(Q_2,Q_1) ~\&~ C(Q_3,Q_4) = C(Q_4,Q_3)]$. 
The proofs of these bounds are similar to the proofs of Theorem \ref{theory:th1} and Theorem \ref{theory:th3}. 

As mentioned in the Results section, the analysis of this manuscript yields tighter LGI bound under the assumption of projective measurements with values ±1. We note that under different assumptions, the LGI parameter can reach the algebraic bound of 4 \cite{budroni-prl2014}. 

    \FloatBarrier
	\bibliographystyle{unsrt}
    \bibliography{qbib}
	
\end{document}